\documentclass[sigconf,natbib,screen=true]{acmart}

\PassOptionsToPackage{dvipsnames}{xcolor}

\usepackage{dsfont}
\usepackage{acronym}
\usepackage{amsmath}
\usepackage[noend]{algorithmic}
\usepackage{algorithm}
\usepackage{bbm}
\usepackage{beramono}
\usepackage{bm}
\usepackage{booktabs}
\usepackage[skip=0pt]{caption}
\usepackage{cleveref}
\usepackage[T1]{fontenc}
\usepackage{graphicx}
\usepackage{hyperref}
\usepackage{listings}
\usepackage{multirow}
\usepackage{natbib}
\usepackage{tabularx}
\usepackage[dvipsnames]{xcolor}
\usepackage{enumerate}
\usepackage[inline]{enumitem}
\usepackage{numprint}
\usepackage[normalem]{ulem}
\usepackage{mleftright}
\usepackage{balance}
\usepackage{makecell}
\usepackage{tabularx}

\setcopyright{rightsretained}

\copyrightyear{2022}
\acmYear{2022}
\setcopyright{rightsretained}
\acmConference[ICTIR '22]{Proceedings of the 2022 ACM SIGIR International Conference on the Theory of Information Retrieval}{July 11--12, 2022}{Madrid, Spain}
\acmBooktitle{Proceedings of the 2022 ACM SIGIR International Conference on the Theory of Information Retrieval (ICTIR '22), July 11--12, 2022, Madrid, Spain}\acmDOI{10.1145/3539813.3545137}
\acmISBN{978-1-4503-9412-3/22/07}

\usepackage{etoolbox}
\makeatletter
\patchcmd{\maketitle}{\@copyrightpermission}{
  \begin{minipage}{0.3\columnwidth}
    \href{http://creativecommons.org/licenses/by/4.0/}{\includegraphics[width=0.90\textwidth]{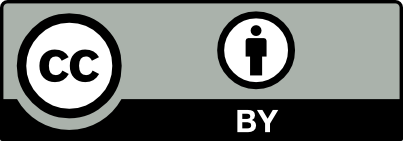}}
  \end{minipage}\hfill
  \begin{minipage}{0.7\columnwidth}
    \href{http://creativecommons.org/licenses/by/4.0/}{This work is licensed under a Creative Commons Attribution International 4.0 License.}
    \end{minipage}
    
   \vspace{5pt}
}{}{}

\makeatother



\definecolor{rj}{RGB}{0, 150, 0}
\definecolor{mdr}{RGB}{200, 0, 0}
\definecolor{ho}{RGB}{0, 50, 150}

\acrodef{IR}{information retrieval}
\acrodef{LTR}{learning to rank}
\acrodef{ARP}{average relevance position}
\acrodef{DCG}{discounted cumulative gain}
\acrodef{EM}{expectation-maximization}
\acrodef{CTR}{click-through-rate}
\acrodef{IPS}{inverse-propensity-scoring}

\acrodef{PL}{Plackett-Luce}

\newcolumntype{Y}{>{\centering\arraybackslash}X}

\allowdisplaybreaks

\author{Harrie Oosterhuis}
\affiliation{%
	\institution{Radboud University}
	\city{Nijmegen}
	\country{The Netherlands}
}
\email{harrie.oosterhuis@ru.nl}

\acmSubmissionID{3}

\title[Reaching the End of Unbiasedness: The Limitations of Click-Based Learning to Rank]{Reaching the End of Unbiasedness:\\ Uncovering~Implicit~Limitations~of~Click-Based~Learning~to~Rank}

\allowdisplaybreaks

\begin{document}

\begin{abstract}
Click-based \ac{LTR} tackles the mismatch between click frequencies on items and their actual relevance.
The approach of previous work has been to assume a model of click behavior and to subsequently introduce a method for unbiasedly estimating preferences under that assumed model.
The success of this approach is evident in that unbiased methods have been found for an increasing number of behavior models and types of bias.

This work aims to uncover the implicit limitations of the high-level prevalent approach in the counterfactual \ac{LTR} field.
Thus, in contrast with limitations that follow from explicit assumptions, our aim is to recognize limitations that the field is currently unaware of.
We do this by inverting the existing approach: we start by capturing existing methods in generic terms, and subsequently, from these generic descriptions we derive the click behavior for which these methods can be unbiased.
Our inverted approach reveals that there are indeed implicit limitations to the counterfactual \ac{LTR} approach:
we find counterfactual estimation can only produce unbiased methods for click behavior based on affine transformations.
In addition, we also recognize previously undiscussed limitations of click-modelling and pairwise approaches to click-based \ac{LTR}.
Our findings reveal that it is impossible for existing approaches to provide unbiasedness guarantees for all plausible click behavior models.
\end{abstract}

\maketitle

\acresetall

\section{Introduction}
\label{sec:intro}

Search and recommendation systems make large collections easily accessible and navigable for users~\citep{qin2013introducing, dato2016fast, Chapelle2011, sorokina2016amazon, gomez2015netflix, covington2016deep}.
A vital component of such systems are ranking models that - based on the features of the items - create short ranked lists of items to present to the user~\citep{liu2009learning, oosterhuis2021computationally, wang2018lambdaloss, burges2010ranknet}.
The experience of users depends on these rankings providing them the items that are relevant to their current needs or match their preferences~\citep{sanderson2010user, sanderson2010, jarvelin2002cumulated, moffat2013users, moffat2017incorporating}.
Ranking systems are applied to an enormous variety of settings, e.g. web-search engines~\citep{qin2013introducing, dato2016fast, Chapelle2011}, e-commerce~\citep{sorokina2016amazon}, video platforms~\citep{gomez2015netflix, covington2016deep}, online advertising~\citep{lagree2016multiple, Komiyama2015}, etc.
As a result of this large variety, what is relevant or preferred also heavily varies over applications and contexts.
Consequently, methods for optimizing ranking systems based on user click behavior have become increasingly popular, since they have the potential to automatically adapt to the exact setting in which they are applied~\citep{joachims2002optimizing, hofmann2013thesis, oosterhuis2020thesis}.

However, click behavior is not an unbiased representation of user preferences or relevance~\citep{sanderson2010user, joachims2017accurately, radlinski2008does}, since there are many other factors that affect whether an item is clicked or not:
i.e.\ items that are not displayed cannot be clicked~\citep{ovaisi2020correcting}, and furthermore, it is well known that the position at which an item is displayed heavily affects the number of clicks it receives~\citep{craswell2008experimental, pan2007google}.
Accordingly, the unbiased \ac{LTR} field has introduced several techniques for dealing with the mismatch between clicks and actual preferences~\citep{joachims2017unbiased, wang2016learning, oosterhuis2020thesis, agarwal2019counterfactual, ai2021unbiased}.
Typically, previous work in this field first identifies certain forms of bias, e.g.\ position-bias~\citep{wang2016learning, joachims2017unbiased, agarwal2019counterfactual, hofmann2014effects}, item-selection-bias~\citep{ovaisi2020correcting, oosterhuis2020topkrankings}, trust-bias~\citep{agarwal2019addressing, vardasbi2020trust, pan2007google, agichtein2006learning}, etc., 
then it assumes a mathematical model of click behavior that describes the identified bias, and subsequently, introduces a method that can provenly optimize a ranking model unbiasedly.
Thus far, this approach has been very successful as methods have been found for an exceeding number of behavior models~\citep{wang2016learning, joachims2017unbiased, oosterhuis2020topkrankings, vardasbi2020trust, oosterhuis2020thesis}.
The large number of recent publications make it evident that the success of this approach is still ongoing and that unbiased \ac{LTR} remains a very active field~\citep{agarwal2019counterfactual, ovaisi2020correcting, oosterhuis2020topkrankings, agarwal2019addressing, vardasbi2020trust, oosterhuis2021onlinecounterltr, ai2018unbiased, hu2019unbiased, vardasbi2021mixture, oosterhuis2020thesis, zhuang2021cross, ai2021unbiased}.

In this work, we aim to investigate the implicit limitations of the prevalent approach in unbiased \ac{LTR}.
In other words, our goal is to better understand what problems we can and cannot expect unbiased \ac{LTR} to solve in the future.
In particular, our focus is on limitations that are implicit, i.e.\ limitations that follow from implicit assumptions in the existing approach that the field is thus currently unaware of.
We believe this is important for two main reasons:
Firstly, by understanding how the current approach is lacking, the field could be better equipped to find alternative approaches in the future.
Secondly, due to its success, the expectation could exist that every click-based \ac{LTR} method should follow the prevalent unbiased \ac{LTR} approach.
However, such an expectation could - unknowingly - systematically exclude problem settings where this approach can never be successful.
For the field, it is thus important to understand when unbiasedness is impossible and change its expectations for future work accordingly.

Our method of finding the limitations of the prevalent approach is to invert it:
we start with generic formulations that describe the main branches of existing methods in the unbiased \ac{LTR} field: counterfactual estimation~\citep{wang2016learning, joachims2017unbiased, oosterhuis2020thesis} and click modelling~\citep{chuklin2015click, vardasbi2021mixture, zhuang2021cross}.
Subsequently, from these generic descriptions, we derive the conditions under which they can be unbiased.
Our process reveals that conditions exists where counterfactual estimation can \emph{never} produce an unbiased method.
Furthermore, we find that click modelling does not seem to provide robust unbiasedness guarantees.
Lastly, we prove that the starting assumptions of the \emph{Unbiased LambdaMART} method~\citep{hu2019unbiased} are problematic since they implicitly assume clicks are unbiased indicators of relevance.
These findings have serious implications for the unbiased \ac{LTR} field:
most importantly, they reveal that either very different novel approaches are needed for unbiased click-based \ac{LTR} in certain settings, or that it may be the case that unbiasedness will never be feasible in certain circumstances.

\section{Preliminaries: The LTR Goal}

Before we start our discussion of click-based \ac{LTR}, this section will explicitly state the assumptions underlying the \ac{LTR} task.
Firstly, \ac{LTR} work generally assumes that each item has a certain merit~\citep{biega2018equity, liu2009learning, morik2020, diazevaluatingstochastic, burges2010ranknet}, often referred to as its \emph{relevance} (in search contexts) or \emph{user preference} (in recommendation contexts).
For the sake of simplicity, we use the \emph{relevance} terminology for the remainder of this work.
In formal terms, relevance is often interpreted as a scalar value; without loss of generality, we define relevance as the probability that a user prefers this item conditioned on the context $q$, i.e.\ for an item $d$:
\begin{equation}
P(R = 1 \mid d, q) = R_{d|q} \in [0,1].
\end{equation}
In a web-search setting, the context $q$ is generally referred to as the \emph{query}, that represents the user's search request so that $R_{d|q}$ can indicate how well a document and the user's information need matches.
Alternatively, in a recommendation setting $q$ may contain \emph{personalized information} about the user from which their preferences could be inferred.
In addition, we allow $q$ to contain any other contextual information that could affect relevance or preference, e.g.\ the time of day/year, preferred language by the user, etc.
Importantly, our definition of $q$ excludes information specific to any item $d$, i.e.\ for each individual ranking $y$, $q$ is the same for all items. 
Our intentionally generic definition of the context $q$ allows our assumptions to cover an enormous variety of ranking settings.
For simplicity, this work will use the \emph{query} terminology.

The goal of the ranking system is to provide a ranking $y$ of items to the user~\citep{liu2009learning, oosterhuis2021computationally, wang2018lambdaloss, burges2010ranknet, xia2008listwise}, $y$ is simply an ordered list of items:
\begin{equation}
y = [d_1, d_2, \ldots].
\end{equation}
The quality of a ranking for a query $q$ is determined by a weighted sum of associated the $R_{d|q}$ values, where the weight per item is determined by its rank in $y$~\citep{jarvelin2002cumulated, moffat2013users, moffat2017incorporating}.
Let the $\text{rank}(d \mid y) \in \mathbb{N}_{>0}$ function indicate the rank of $d$ in $y$ and $\lambda: \mathbb{N}_{>0} \rightarrow \mathbb{R}$ the weight function, the quality of ranking is then:
\begin{equation}
\mathcal{R}_{q}(y) = \sum_{d \in y} \lambda(\text{rank}(d \mid y)) R_{d|q}.
\label{eq:rankingquality}
\end{equation}
In general, ranking models predict a score for each item conditioned on the query and a ranking is created by sorting items according to their predicted scores~\citep{liu2009learning, oosterhuis2021computationally, wang2018lambdaloss, burges2010ranknet, xia2008listwise}.
In formal terms, the quality of a ranking model $s$ is the expected quality of its rankings over the natural distribution of the queries:
\begin{equation}
\label{eq:existingunbiassystem}
\mathcal{R}(s) = \mathbb{E}_{q}[\mathcal{R}_{q}(y_{q,s})].
\end{equation}
In practice this expectation is often estimated based on a dataset that contains a large sample of queries, i.e.\ web-search datasets that contain many queries logged from users~\citep{qin2013introducing, dato2016fast, Chapelle2011}.

This concludes our very generic description of the \ac{LTR} goal and its explicit assumptions.
We note that even this generic description brings several limitations, for instance, its main assumption is that ranking quality can be decomposed into rank-weights and item relevances that only depend on query-item combinations (Eq.~\ref{eq:rankingquality}).
Consequently, this explicitly excludes any notion of diversity in a ranking, e.g.\ it assumes users do not mind if all items in a ranking are very similar or identical~\citep{xia2017adapting, hurley2011novelty, clarke2008novelty}.
We will not further discuss this particular limitation since it follows from explicit assumptions, and it is well-known and discussed in previous work~\citep{rodrygo2015search, dang2012diversity, radlinski2006improving}.

\section{Existing Click-Based LTR Approach}
\label{sec:background}

This section will describe the two main families of methods in click-based \ac{LTR}: counterfactual estimation and click modelling.

\subsection{Counterfactual Estimation}
\label{sec:existingcounterfactual}

Counterfactual estimation appears to be the most common approach in click-based \ac{LTR} literature~\citep{wang2016learning, joachims2017unbiased, oosterhuis2020topkrankings, vardasbi2020trust, oosterhuis2020thesis, agarwal2019counterfactual, agarwal2019addressing, oosterhuis2021onlinecounterltr, ai2018unbiased, hu2019unbiased, wu2021unbiased}.
This approach starts by assuming a probabilistic model of click behavior, and subsequently, derives an estimation method for estimating the relevance of an item based on clicks.
Mathematically, the method can usually be proven to provide unbiased estimates of the relevances.
To illustrate this approach, we will briefly describe earlier work by \citet{vardasbi2020trust} for affine click models~\citep{agarwal2019addressing}.

Let $k$ denote a display position where an item can be displayed, per position there are the parameters
$\alpha_{k} \in [0,1]$ and $\beta_{k} \in [0,1]$ so that $\alpha_{k} + \beta_{k} \in [0,1]$.
The affine model assumes that the probability of a click $C \in \{0,1\}$ at position $k$ is an affine transformation of the item's relevance:
\begin{equation}
P(C = 1 \mid d, k, q) = \alpha_{k} R_{d|q} + \beta_{k}.
\label{eq:affineclickmodel}
\end{equation}
This model was originally proposed by \citet{agarwal2019addressing} to capture both position-bias and trust-bias. %
In the earliest unbiased \ac{LTR} work a linear transformation was assumed that only describes position-bias~\citep{wang2016learning, joachims2017unbiased, agarwal2019counterfactual, ai2018unbiased}, in the affine model this scenario can be captured with: $\forall k, \, \beta_{k} = 0$.
In the linear case, $\alpha_{k}$ is called a propensity, when $\beta_{k} \not= 0$, $\alpha_{k}$ can be interpreted as the correlation between relevance and clicks at position $k$.
Non-zero $\beta_{k}$ allow the model to capture trust-bias: an increase in clicks regardless of relevance at higher ranks due to the trust users have in the ranking system~\citep{agarwal2019addressing, vardasbi2020trust, vardasbi2021mixture, oosterhuis2021onlinecounterltr, joachims2017accurately, pan2007google, agichtein2006learning}.

Clicks under this affine model are biased in that the expected \acp{CTR} are not equal to relevances.
Let $\mathcal{D}_q$ indicate a set of observed interactions with query $q$, where $c_i(d) \in \{0,1\}$ indicates whether item $d$ was clicked at interaction $i$ and $k_i(d)$ the position at which $d$ was displayed at interaction $i$.
We can then express the bias of \acp{CTR} formally:
\begin{equation*}
\mathbb{E}_{c,k}\bigg[
\frac{1}{|\mathcal{D}_q|} \sum_{i \in \mathcal{D}_q} \mleft(c_i(d)
- R_{d|q}\mright)\bigg]
=
\mathbb{E}_{k}\big[(\alpha_{k(d)} - 1) R_{d|q} + \beta_{k(d)}\big] \not= 0.
\end{equation*}
Depending on the $\alpha$ and $\beta$ parameters and the positions where $d$ was displayed, its \ac{CTR} will over or underestimate its relevance.
Naively ranking items according to their \acp{CTR} will give the items displayed at higher ranks an unfair advantage due to position-bias and trust-bias~\citep{joachims2017unbiased, agarwal2019addressing, joachims2017accurately, craswell2008experimental}.
In other words, high \acp{CTR} are often more indicative of how an item was displayed than its relevance.

Based on an inversion of Eq.~\ref{eq:affineclickmodel}, \citet{vardasbi2020trust} introduced an affine estimator that uses the estimated $\hat{\alpha}$ and $\hat{\beta}$ parameters:
\begin{equation}
\hat{R}_{d|q} = \frac{1}{|\mathcal{D}_q|} \sum_{i \in \mathcal{D}_q}
\frac{c_i(d) - \hat{\beta}_{k_i(d)}}{\hat{\alpha}_{k_i(d)}},
\label{eq:existingcfest}
\end{equation}
that can easily be used to estimate the quality of a ranking:
\begin{equation}
\widehat{\mathcal{R}}_{q}(y) = \sum_{d \in y} \lambda(\text{rank}(d \mid y)) \hat{R}_{d|q}.
\end{equation}
The affine estimator has the following bias:
\begin{equation*}
\mathbb{E}_{c,k}\mleft[
\hat{R}_{d|q}
- {R}_{d|q}
\mright] = 
\mathbb{E}_{k}\mleft[
\frac{\big(\alpha_{k(d)} - \hat{\alpha}_{k(d)}\big) R_{d|q} + \beta_{k(d)} - \hat{\beta}_{k(d)}}{\hat{\alpha}_{k(d)}}
\mright].
\end{equation*}
It is easy to see that this estimator is unbiased when the estimated bias parameters are correct~\citep{vardasbi2020trust, oosterhuis2021onlinecounterltr, vardasbi2021mixture}.
Moreover, it can then also be used for an unbiased estimate of ranking quality:
\begin{equation}
\begin{split}
&\mleft(
\forall d \in y, \,
\forall i \in \mathcal{D}_q,\, \hat{\alpha}_{k_i(d)} = \alpha_{k_i(d)} \land 
\hat{\beta}_{k_i(d)} = \beta_{k_i(d)} \mright)
\\ &\hspace{0.36cm}
\longrightarrow
\forall d \in y, \,
\mathbb{E}_{c,k}\mleft[
\hat{R}_{d|q}\mright] = R_{d|q}
\longrightarrow
\mathbb{E}_{c,k}\mleft[
\widehat{\mathcal{R}}_{q}(y)\mright] 
= \mathcal{R}_{q}(y).
\end{split}
\end{equation}
Trivially, this can provide an unbiased estimate of model performance $\mathcal{R}(s)$ which in turn can be optimized with standard \ac{LTR} techniques~\citep{joachims2017unbiased, oosterhuis2020topkrankings, agarwal2019counterfactual, wang2018lambdaloss}.
Therefore, \citet{vardasbi2020trust} provide an unbiased \ac{LTR} method for their assumed affine click model.

The above example illustrates how \citet{vardasbi2020trust} applied the common counterfactual approach to unbiased \ac{LTR}, other work has applied it to many different models of click behavior:
For example, early work applied this approach to position-bias only~\citep{wang2016learning, joachims2017unbiased, agarwal2019counterfactual, ai2018unbiased};
\citet{fang2019intervention} to position-bias that varies per query;
\citet{oosterhuis2020topkrankings} to position-bias and item-selection-bias;
\citet{vardasbi2020cascade} to cascading examination behavior;
\citet{wu2021unbiased} to a bias from surrounding items; and \citet{oosterhuis2021onlinecounterltr} address position-bias, trust-bias and item-selection-bias all at once.

\subsection{Click Modelling}
\label{sec:clickmodel}
The other main approach in click-based \ac{LTR} is click modelling, this approach is both used to make relevance estimates~\citep{chuklin2015click, borisov2016neural, zhuang2021cross, chen2020context, chapelle2009dynamic, borisov2016context, borisov2018click} as well as estimating the bias parameters values for counterfactual estimation~\citep{wang2016learning, vardasbi2020trust, oosterhuis2020taking, wang2018position}.
The approach is to fit a generative model of click behavior on the observed click data, the hope is that the resulting fitted model captures both the bias and relevances by modelling both simultaneously~\citep{chuklin2015click}.
Similar to the counterfactual estimation, the click modelling approach starts by assuming a probabilistic model of click behavior.
For example, we can again assume the affine click model in Eq.~\ref{eq:affineclickmodel}~\citep{agarwal2019addressing, vardasbi2020trust}, the predictive model we want to fit would thus be:
\begin{equation}
\widehat{P}(C = 1 \mid d, k, q) = \hat{\alpha}_{k} \hat{R}_{d|q} +  \hat{\beta}_{k}.
\end{equation}
Let $\hat{A}$ and $\hat{B}$ indicate the set of estimated $\hat{\alpha}$ and $\hat{\beta}$ parameters and $\hat{R}_q$ a vector of relevance estimates of all items $\hat{R}_{d|q}$ for query $q$.
The model can be fitted to data using the negative log-likelihood loss:
\begin{equation}
\begin{split}
&\mathcal{L}(\hat{R}_q, \hat{A},  \hat{B}, \mathcal{D}_q)
=
\frac{-1}{|\mathcal{D}_q||\hat{R}_q|}\sum_{i \in \mathcal{D}_q}\sum_{d \in \hat{R}_q}
c_i(d) \log\big(\hat{\alpha}_{k_i(d)}\hat{R}_{d|q}
\\&\hspace{0.7cm}
+ \hat{\beta}_{k_i(d)}\big)
+ \big(1 - c_i(d)\big) \log\big(1 - \big(\hat{\alpha}_{k_i(d)}\hat{R}_{d|q} + \hat{\beta}_{k_i(d)}\big)\big).
\end{split}
\end{equation}
An important difference with the counterfactual estimation approach is that the relevance estimates $\hat{R}_q$ are not the output of an estimator function but parameters to be optimized.
If we consider the expected value of the loss, we see that in expectation it is equal to a negative log-likelihood loss between the predictive click model and the actual click probabilities:
\begin{align}
&\mathbb{E}_{c,k}\mleft[\mathcal{L}(\hat{R}_q, \hat{A},  \hat{B}, \mathcal{D}_q)\mright]
 \\[-0.5ex]&=
\mathbb{E}_{k}\Big[ \frac{-1}{|\hat{R}_q|}\sum_{d \in \hat{R}_q}
(\alpha_{k(d)}R_{d|q} + \beta_{k(d)}) \log\big(\hat{\alpha}_{k(d)}\hat{R}_{d|q} + \hat{\beta}_{k(d)}\big)
\nonumber\\[-1ex]&\hspace{0.4cm}
+ \big(1 - \big({\alpha}_{k(d)}{R}_{d|q} + {\beta}_{k(d)}\big)\big) \log\big(1 - \big(\hat{\alpha}_{k(d)}\hat{R}_{d|q} + \hat{\beta}_{k(d)}\big)\big) \big]
\nonumber
\\[-0.5ex]
&= 
\frac{-1}{|\hat{R}_q|}\sum_{d \in \hat{R}_q}
\mathbb{E}_{k}\big[ P(C = 1 \mid d, k, q) \log\big( \widehat{P}(C = 1 \mid d, k, q)\big)
\nonumber\\[-1.2ex] \nonumber&\hspace{1.5cm}
+
\big(1-P(C = 1 \mid d, k, q)\big) \log \big(1-\widehat{P}(C = 1 \mid d, k, q)\big)\big].
\end{align}
The hope is that by minimizing the loss, the model learns the relevances and bias parameters correctly, since correct parameters would minimize the loss in expectation:
\begin{align}
&
\mleft(
\mleft(
\forall d \in \hat{R}_q, \, \hat{R}_{d|q} = R_{d|q} \mright) \land
\mleft(
\forall k,\, \hat{\alpha}_{k} = \alpha_{k} \land 
\hat{\beta}_{k} = \beta_{k}
 \mright) \mright)
\\
&\;
\longrightarrow
\mathbb{E}_{c,k}\mleft[\mathcal{L}(\hat{R}_q, \hat{A},  \hat{B}, \mathcal{D}_q)\mright]
= \min_{\hat{R}_q', \hat{A}',  \hat{B}'} \mathbb{E}_{c,k}\mleft[\mathcal{L}(\hat{R}_q', \hat{A}',  \hat{B}', \mathcal{D}_q)\mright].
\nonumber
\end{align}
However, it is important to realize the direction of the implication which reveals that this is not a strong guarantee: minimizing the expected loss does not imply that the correct parameters have been found.
Section~\ref{sec:limitclickmodel} will take a critical look at the unbiasedness gaurantees this approach can provide.

Click modelling was initially - and could still be - seen as its own field within \acl{IR}: early methods used intuitive graphical models~\citep{chuklin2015click, chapelle2009dynamic, borisov2016context}, later work relies on neural models that are hard to interpret but have better predictive performance~\citep{borisov2016neural, chen2020context, borisov2018click}.
With the inception of unbiased \ac{LTR}, click models became a common choice for estimating bias parameters for counterfactual estimation instead of relevance estimation~\citep{wang2016learning, vardasbi2020trust, oosterhuis2020taking, wang2018position}.
However, recently novel click models for relevance estimation have been introduced again:
For example, \citet{vardasbi2021mixture} proposed a click model that uses a prior distribution over relevances; and \citet{zhuang2021cross} used a click model with a state-of-the-art deep learning architecture for relevance estimation in a grid layout.
They found their click model to be more effective than existing counterfactual estimation techniques in a real-world recommendation setting.

\subsection{Shortcomings of the Existing Approach}
\label{sec:shortcomings}

The unbiased \ac{LTR} field has been very successful both in terms of theoretical results~\citep{oosterhuis2020thesis, oosterhuis2021onlinecounterltr, joachims2017unbiased, morik2020} as in performance improvements for real-world search and recommendation systems~\citep{wang2016learning, joachims2017unbiased, zhuang2021cross, agarwal2019addressing}.
However, there are also some well-known shortcomings with the existing approach that should be discussed:

First, while in theory counterfactual estimation can provide provenly unbiased estimators, in practice the application of these estimators is rarely actually unbiased.
There are three main reasons for this mismatch:
\begin {enumerate*} [label=(\roman*)]
\item the assumed behavior model is incorrect - real-world user behavior is rarely captured perfectly by mathematical models;
\item the propensities are inaccurate - perfect propensity estimation is infeasible in practice~\citep{wang2018position, agarwal2019estimating};
and most importantly \item it is standard to apply clipping in practice - propensities are clipped by a threshold $\tau \in [0,1]$: 
$\hat{\alpha}_{k(d)}^\text{clip} = \max(\hat{\alpha}_{k(d)}, \tau)$, this introduces some bias but can greatly reduce variance~\citep{joachims2017unbiased, strehl2010logged}.
\end{enumerate*}
The first two points are understandable: models and parameters are often inaccurate and small inaccuracies are often inconsequential.
However, clipping generally has a large effect and purposefully introduces notable bias to the estimation process.
It is at least surprising that most work in the unbiased \ac{LTR} field puts much emphasis on the unbiased estimators they introduce, while simultaneously applying biased versions of the estimators in their experiments~\citep{joachims2017unbiased, agarwal2019addressing, oosterhuis2021onlinecounterltr, vardasbi2020trust, oosterhuis2020topkrankings}.

Second, while click modelling has seen real-world success~\citep{zhuang2021cross}, there are very little guarantees regarding its performance.
In stark contrast with counterfactual estimation - and to the best of our knowledge - it is currently unclear whether there are broad  theoretical guarantees for when click models can provide correct relevance estimates.
Even when the structure of the assumed behavior model is correct, that does not guarantee that a fitted click model will have unbiased relevance estimates.
From a theoretical perspective, it is thus unclear in what circumstances click models are reliable or not.

Lastly, while there have been several high-level looks at the unbiased \ac{LTR} field, they have mainly compared the subgroups of online and counterfactual \ac{LTR} methods~\citep{ai2021unbiased, jagerman2019comparison, oosterhuis2020unbiased}.
To the best of our knowledge, previous work has not looked at the properties of the high-level approach of the Unbiased \ac{LTR} field, instead of its individual methods.
It is thus currently not known what the limitations of the prevalent approach are, nor what they may entail.

\section{Methodology}
\label{sec:method}

Section~\ref{sec:background} described the prevalent approaches in the unbiased \ac{LTR} field and noted that the limitations of the high-level approach are currently unclear.
In this work, we will investigate whether there are implicit assumptions in the two main families of debiasing approaches.
In particular, we want to find out whether they impose limits on which circumstances estimation methods can or cannot be unbiased or consistent. 
This section will explain our methodology to investigating such implicit limitations, we will start by defining the two theoretical properties that we are interested in.

We begin by formally defining unbiasedness:
\begin{definition}
\label{def:unbiasedness}
\emph{Unbiasedness.}
A click-based relevance estimation method is \emph{unbiased}, if the expected values of all the resulting relevance estimates are equal to the true relevances:
\begin{equation}
\forall (d, q),\;
\mathbb{E}\big[ \hat{R}_{d|q} \big] = R_{d|q}.
\end{equation}
\end{definition}
Our definition is slightly different than the standard definition in the previous literature:
where existing work has focussed on the estimated performance of a ranking system (Eq.~\ref{eq:existingunbiassystem})~\citep{joachims2017unbiased, oosterhuis2020thesis, agarwal2019counterfactual}, we focus on bias of the individual relevance estimates.
We motivate this difference with three arguments:
\begin {enumerate*} [label=(\roman*)]
\item Our definition based on individual relevance estimates practically guarantees an unbiased performance estimate (c.f.\ Eq.~\ref{eq:rankingquality});
\item To the best of our knowledge, there is no existing work that provides an unbiased performance estimate without also providing unbiased relevance estimates~\citep{wang2016learning, joachims2017unbiased, oosterhuis2020thesis, agarwal2019counterfactual, oosterhuis2020topkrankings, vardasbi2020trust, oosterhuis2021onlinecounterltr, ai2018unbiased, hu2019unbiased, wu2021unbiased, vardasbi2020cascade}.
In other words, all existing unbiased \ac{LTR} methods are also unbiased under our definition.
\item Our definition is simpler and easier to work with.
\end{enumerate*}
To summarize, we define unbiasedness based on individual relevance estimates because it is easier to work with and is interchangeable with the unbiased definition based on system performance in previous work.

While unbiasedness has been the main focus of most click-based \ac{LTR} work, it is not the only property that is important to the field.
As Section~\ref{sec:shortcomings} noted, in practice counterfactual estimation methods are often deployed in a biased manner.
In particular, propensity clipping introduces some bias but reduces variance greatly and is widely applied in counterfactual \ac{LTR}~\citep{joachims2017unbiased, agarwal2019addressing, oosterhuis2021onlinecounterltr, vardasbi2020trust, oosterhuis2020topkrankings}.
Similarly, it is conceptually hard to apply the unbiasedness property to click modelling methods~\citep{vardasbi2021mixture}.
For these reasons, we will also consider \emph{consistency}~\citep{oosterhuis2020taking} with the following formal definition:
\begin{definition}
\label{def:consistency}
\emph{Consistency.}
A click-based relevance estimation method is \emph{consistent}, if the values of all the resulting relevance estimates are equal to the true relevances in the limit of an infinite number of interactions:
\begin{equation}
\lim_{|\mathcal{D}_q| \to \infty} \hat{R}_{d|q} = R_{d|q}.
\end{equation}
\end{definition}
Consistency is a desirable property since it guarantees accurate convergence as interaction data continues to increase.
Furthermore, variance is generally negligible when the number of interaction is extremely large~\citep{joachims2017unbiased}, therefore unlike unbiasedness, there is often not a trade-off between variance and consistency.

Our goal is thus to identify under which conditions the unbiased \ac{LTR} approach can provide methods that are unbiased and consistent according to our definitions.
Our method to identifying these conditions inverts the prevalent approach that starts with behavior assumptions and derives an unbiased method.
In contrast, we first describe the two main families: counterfactual estimation and click modelling, in the most generic terms.
Subsequently, we derive click behavior conditions from these generic methods, thereby revealing the assumptions that are implicitly present in the existing approach.
The following two sections will do this for counterfactual estimation and click modelling, in addition, Section~\ref{sec:othermethods} will discuss other methods that fall outside the former two categories.

\section{The Limitations of Click-Based Counterfactual Estimators}
\label{sec:limitcount}

As discussed in Section~\ref{sec:existingcounterfactual}, counterfactual estimation represents the largest branch of the unbiased \ac{LTR} field~\citep{wang2016learning, joachims2017unbiased, oosterhuis2020thesis}.
The underlying assumption of these methods is that click probabilities are decomposable into relevance and display factors, e.g.\ the estimator in Eq.~\ref{eq:existingcfest} assumes click probabilities are determined by relevance and display-position according to Eq.~\ref{eq:affineclickmodel}~\citep{vardasbi2020trust}.
Click-based counterfactual estimators aim to convert a click signal into an unbiased relevance signal by correcting for the display factors.
In order to investigate the entire family of click-based counterfactual estimators, we first define a display context $x(d)$ in the most generic and broad terms:
\begin{definition}
\label{def:displaycontext}
\emph{Display Context.}
The display context $x_i(d)$ contains all information about how item $d$ is displayed that could affect the click probability of $d$ at interaction $i$.
It does not contain any information about the relevance $d$: $R_{d|q}$.
In other words,
one should be able to determine $x_i(d)$ without any knowledge of $R_{d|q}$.
\end{definition}
For example, for the estimator in Eq.~\ref{eq:existingcfest} the display context is the display position: $x_i(d) = k_i(d)$~\citep{agarwal2019addressing}. 
Many alternatives are possible, e.g.\ $x_i(d)$ can also represent a probability distribution over positions as in the policy-aware estimator ~\citep{oosterhuis2020topkrankings, oosterhuis2021onlinecounterltr}.
With this broad definition of the display context, we propose a generic description of a click-based counterfactual estimator:
\begin{definition}
\label{def:counterfactualestimate}
\emph{Counterfactual Relevance Estimate.}
A click-based counterfactual relevance estimate is an average over independently sampled interactions where each click or non-click is transformed by a function $f$ such that:
\begin{equation}
\hat{R}_{d|q} = \frac{1}{N_q}\sum_{i \in \mathcal{D}_q} f(c_i(d), x_{i}(d)).
\label{eq:counterfactualestimate}
\end{equation}
Accordingly, $f$ only has two relevant values per context $x(d)$ for the counterfactual estimate:
$f(1, x(d))$ when a click takes place on $d$ and $f(0, x(d))$ when no click takes place.
\end{definition}
To the best of our knowledge, Definition~\ref{def:counterfactualestimate} covers all existing click-based counterfactual estimators~\citep{wang2016learning, joachims2017unbiased, oosterhuis2020topkrankings, vardasbi2020trust, oosterhuis2020thesis, agarwal2019counterfactual, agarwal2019addressing, oosterhuis2021onlinecounterltr, ai2018unbiased, wu2021unbiased, vardasbi2020cascade} with the exception of three counterfactual pairwise estimators~\citep{wang2021non, hu2019unbiased, saito2020unbiased} that will be discussed in Section~\ref{sec:othermethods}.
For example, we see that the estimator in Eq.~\ref{eq:existingcfest} is a specific instance where $f$ is chosen so that $f(1,x_i(d)) = \frac{1- \beta_{k_i(d)}}{\alpha_{k_i(d)}}$ and $f(0,x_i(d)) = \frac{- \beta_{k_i(d)}}{\alpha_{k_i(d)}}$.
The aim of our definition is to cover all existing estimators and as many future estimators as possible.
Accordingly, - to the best of our knowledge - different choices of $f$ can cover almost all existing counterfactual estimators, and moreover, they can also cover many more possible estimators that have yet to be introduced to the field.

With these definitions, we can now start to derive the conditions for which a click-based counterfactual estimator is unbiased or consistent.
First, we note that the expected value of a relevance estimate $\hat{R}_{d|q}$ is simply the expected value of $f$ conditioned on $q$:
\begin{lemma}
\label{lemma:cfest:expected}
In expectation a counterfactual relevance estimate is:
\begin{equation}
\mathbb{E}_{c,x}\mleft[ \hat{R}_{d|q} \mright] 
= \mathbb{E}_{c,x}\mleft[ f(c(d), x(d)) \mid q \mright].
\end{equation}
\end{lemma}
\begin{proof} Using Definition~\ref{def:counterfactualestimate}:
\begin{align}
&\mathbb{E}_{c,x}\mleft[ \hat{R}_{d|q} \mright]
= \mathbb{E}_{c,x}\bigg[ \frac{1}{N_q}\sum_{i \in \mathcal{D}_q} f(c_i(d), x_{i}(d)) \bigg]
\\
&= \mathbb{E}_{c,x}\mleft[ f(c_i(d), x_{i}(d)) \mid i \in D_q \mright]
= \mathbb{E}_{c,x}\mleft[ f(c(d), x(d)) \mid q \mright]. \qedhere
\end{align}
\end{proof}
With this Lemma, we can prove that an unbiased estimator is always consistent and vice-versa:
\begin{theorem}
\label{theorem:cfbiasandcons}
A counterfactual estimator is consistent if and only if it is unbiased:
\begin{equation}
\mathbb{E}_{c,x}\mleft[ \hat{R}_{d|q} \mright] = R_{d|q} \longleftrightarrow \lim_{|\mathcal{D}_q| \to \infty} \hat{R}_{d|q} = R_{d|q}.
\end{equation}
\end{theorem}
\begin{proof}
With the use of Definition~\ref{def:counterfactualestimate} and Lemma~\ref{lemma:cfest:expected}:
\begin{equation}
\lim_{|\mathcal{D}_q| \to \infty} \hat{R}_{d|q}
= \mathbb{E}_{c,x}\mleft[f(c_i(d), x_{i}(d)) \mid q \mright]
= \mathbb{E}_{c,x}\mleft[ \hat{R}_{d|q} \mright].
\qedhere
\end{equation}
\end{proof}
Therefore, it appears that the focus of previous work to prove unbiasedness was not incorrectly placed, since implicitly this has also proved consistency.
Nevertheless, we think it is important that we have a theoretical motivation for this focus now.
Finally, we can derive the following conditions for the unbiasedness - and therefore also the consistency - of click-based counterfactual estimators:
\begin{theorem}
\label{theorem:counterestimateunbiased}
\label{theorem:cf}
A counterfactual estimator can only be unbiased if click probabilities follow an affine transformation of relevance s.t:
\begin{equation}
\label{eq:counter:clickprobaffine}
R_{d|q} =
\mathbb{E}_{x}\mleft[ 
\frac{P(C =1 \mid d, x, q) - \beta_{x(d)}}{\alpha_{x(d)}}
\mid q\mright].
\end{equation}
\end{theorem}
\begin{proof}
From Lemma~\ref{lemma:cfest:expected} it follows that the unbiasedness criteria for a counterfactual relevance estimate can be reformulated as:
\begin{equation}
\mathbb{E}_{c,x}\mleft[ \hat{R}_{d|q} \mright] = R_{d|q} \longleftrightarrow
\mathbb{E}_{c,x}\mleft[ f(c(d), x({d})) \mid q \mright] = R_{d|q}.
\label{eq:countercrit1}
\end{equation}
The expected value can be rewritten to:
\begin{align}
&\mathbb{E}_{c,x}\mleft[ f(c(d), x({d})) \mid q \mright]
\label{eq:expectrewrit}\\&
=
\mathbb{E}_{x}[ 
P(C =1 \mid d, x, q)f(1, x({d}))
\nonumber\\
&\hspace{3.3cm} + (1 - P(C =1 \mid d, x, q))f(0, x({d}))
\mid q]
\nonumber\\&
=
\mathbb{E}_{x}\mleft[ 
P(C =1 \mid d, x, q)(f(1, x({d})) - f(0, x({d}))) + f(0, x({d}))
\mid q\mright].
\nonumber
\end{align}
Combining Eq.~\ref{eq:countercrit1} and~\ref{eq:expectrewrit} reveals that the unbiasedness criteria also has implications on the click probability: 
\begin{align}
\mathbb{E}_{c,x}\big[ \hat{R}_{d|q} \big] = R_{d|q} \longleftrightarrow
\mathbb{E}_{x}&\big[
P(C =1 \mid d, x, q)(f(1, x({d}))
\label{eq:theorem55lasteq}
\\&
- f(0, x({d}))) + f(0, x({d}))
\mid q\big]
=
R_{d|q}.
\nonumber
\end{align}
To prove Theorem~\ref{theorem:counterestimateunbiased}, we derive the following $\alpha_{x(d)}$ and $\beta_{x(d)}$  from Eq.~\ref{eq:theorem55lasteq} which
show that click probability is an affine transformation of the form stated in Eq.~\ref{eq:counter:clickprobaffine}:
\begin{align}
\alpha_{x(d)} &= (f(1, x({d})) - f(0, x({d})))^{-1},
\label{eq:alphabetavalueproof}
\\
\beta_{x(d)} &= -f(0, x({d}))(f(1, x({d})) - f(0, x({d})))^{-1}.\qedhere
\end{align}
\end{proof}
Theorem~\ref{theorem:cf} can be difficult to interpret since the expectation over the display context in Eq.~\ref{eq:counter:clickprobaffine} depends on the choice of logging policy.
Nonetheless, Theorem~\ref{theorem:cf} clearly shows that - in spite of being broad and generic - Definition~\ref{def:counterfactualestimate} implicitly limits the click models counterfactual relevance estimation can be unbiased or consistent for.
It appears that this is a result of the estimate being a mean over individually transformed clicks (Eq.~\ref{eq:counterfactualestimate}) which makes its expected value a linear interpolation between the possible $f(0, x(d))$ and $f(1, x(d))$ values.
Consequently, an counterfactual estimator following Definition~\ref{def:counterfactualestimate} can only be unbiased or consistent if clicks follow a transformation that can be corrected by such an interpolation, Theorem~\ref{theorem:cf} proves that such a transformation has to match the affine form of Eq.~\ref{eq:counter:clickprobaffine}.

To better understand what Theorem~\ref{theorem:cf} entails, we consider what it would mean for a deterministic logging policy:
\begin{corollary}
\label{corollary:cfnorandom}
If an item is displayed without randomization then a counterfactual estimator can only be unbiased if the item's click probability is an affine transformation of relevance in the form:
\begin{equation}
P(x(d) \mid q) = 1 \rightarrow 
P(C =1 \mid d, x, q) = \alpha_{x(d)} R_{d|q} + \beta_{x(d)}.
\label{eq:counter:clickprobaffinenorandom}
\end{equation}
\end{corollary}
\begin{proof}
This follows directly from Theorem~\ref{theorem:cf}.
\end{proof}
Furthermore, often practitioners have no control over the exact logging policy that gathers data.
In such a scenario, unbiasedness guarantees cannot rely on a specific choice of logging policy as the deployment of that policy may not be possible:
\begin{corollary}
\label{corollary:cfstrong}
Without control over the logging policy, unbiasedness guarantees for a counterfactual estimator are only possible if click probabilities are affine transformations of relevance in the form:
\begin{equation}
 \forall (d, x, q), \;
P(C =1 \mid d, x, q) = \alpha_{x(d)} R_{d|q} + \beta_{x(d)}.
\label{eq:counter:clickprobaffinestrong}
\end{equation}
\end{corollary}
\begin{proof}
We cannot rule out the possibility that for each possible display context an item exists which the logging policy will only display in that context, thus, Cor.~\ref{corollary:cfstrong} follows from Cor.~\ref{corollary:cfnorandom}.
\end{proof}
While the individual steps of our theoretical derivation are very straightforward, the result is quite significant:
we have proven that the form of click-based counterfactual estimators entail implicit assumptions about the click behavior they are able to unbias.
In other words, because click-based counterfactual estimators transform individual click or non-click signals and then average the result, they can only correct for click probabilities that are affine transformations of relevances.
Conversely, this proves that no unbiased counterfactual relevance estimator (Def.~\ref{def:counterfactualestimate}) can exist for click behavior that does not match the affine transformation of Eq.~\ref{eq:counter:clickprobaffine}.

Furthermore, there is a less obvious implicit requirement that $\alpha_{x(d)}$ and $\beta_{x(d)}$ should only depend on the display context and not on the item relevance:
\begin{corollary}
\label{corollary:cf}
If two items are always displayed in the same single display context, their $\alpha_{x(d)}$ and $\beta_{x(d)}$ values should be equal:
\begin{equation}
\begin{split}
&\big(\exists x', \, P(x(d) = x' \mid q ) = 1 \land P(x(d') = x' \mid q ) = 1 \big) 
\\&\hspace{1.2cm} \longrightarrow
P\big(\alpha_{x(d)} = \alpha_{x(d')} \land \beta_{x(d)} = \beta_{x(d')} \mid q \big) = 1.
\end{split}
\end{equation}
\end{corollary}
\begin{proof} 
According to Definition~\ref{def:displaycontext}, $x(d)$ is not determined by the relevance of $d$.
Therefore, two items with the same display context $x(d)$ should also have the same values from $f$:
\begin{equation}
\begin{split}
&\big(x(d) = x(d')\big) \leftrightarrow \big(f(0, x(d)) = f(0, x(d'))
\\
&\hspace{3.1cm} \land f(1, x(d)) = f(1, x(d'))\big).
\end{split}
\end{equation}
Cor.~\ref{corollary:cfnorandom} and
Eq.~\ref{eq:alphabetavalueproof} show $\alpha_{x(d)}$ and $\beta_{x(d)}$ must then be equal.
\end{proof}
This requirement becomes more obvious when one considers the implications of its negation: a setting where the value of $\alpha_{x(d)}$ and $\beta_{x(d)}$ are determined by $R_{d|q}$.
In such a setting the value of $f(0, x(d))$ and $f(1, x(d))$ are also determined by $R_{d|q}$, and therefore, one should first know the relevance of an item before being able to unbiasedly estimate it.
In other words, in such a setting one can only unbiasedly estimate relevances if one already knows those relevances and thus has no need for estimation.
Importantly, we are not arguing that such settings do not exist but merely that unbiased counterfactual estimation is infeasible in such cases.

In addition to the requirement stated in Theorem~\ref{theorem:cf} and Corollaries~\ref{corollary:cfnorandom}, \ref{corollary:cfstrong} and~\ref{corollary:cf}, there could be more requirements that have yet to be proven.
In particular, our analysis has not considered under what conditions the values of $f$ or $\alpha$ and $\beta$ can be determined.
Previous work on position-bias estimation~\citep{wang2018position, agarwal2019estimating, craswell2008experimental, chuklin2015click} shows that bias estimation is far from a trivial task.
Therefore, it seems reasonable that it is also necessary that the bias parameters $\alpha$ and $\beta$ can be accurately inferred  for unbiased counterfactual estimation.
However, we currently lack a theoretical approach to prove broad conditions about bias estimation in general (Section~\ref{sec:biasestimationcm} will discuss bias estimation with click modelling).

To summarize our theoretical findings on counterfactual estimation for click-based \ac{LTR}:
We have provided a generic definition of click-based counterfactual estimators that captures almost all methods in the field~\citep{wang2016learning, joachims2017unbiased, oosterhuis2020topkrankings, vardasbi2020trust, oosterhuis2020thesis, agarwal2019counterfactual, agarwal2019addressing, oosterhuis2021onlinecounterltr, ai2018unbiased, wu2021unbiased, vardasbi2020cascade}.
From this definition, we have shown that estimators of this form can only provide unbiased and consistent estimates when click probabilities are affine transformations of relevances.
We have thus identified a significant limitation of the high-level approach for click-based counterfactual estimation: it can never find unbiased or consistent methods for non-affine click behavior.

The remainder of this section will provide some example scenarios that illustrate this limitation and discuss some edge-cases related to our definitions.

\begin{figure}[tb]
\centering
{\renewcommand*{\arraystretch}{0.4}
\begin{tabular}{r c}
\rotatebox[origin=c]{90}{
\hspace{-1cm}
\small
Click Probability
\hspace{-4cm}
}
& \includegraphics[scale=0.475]{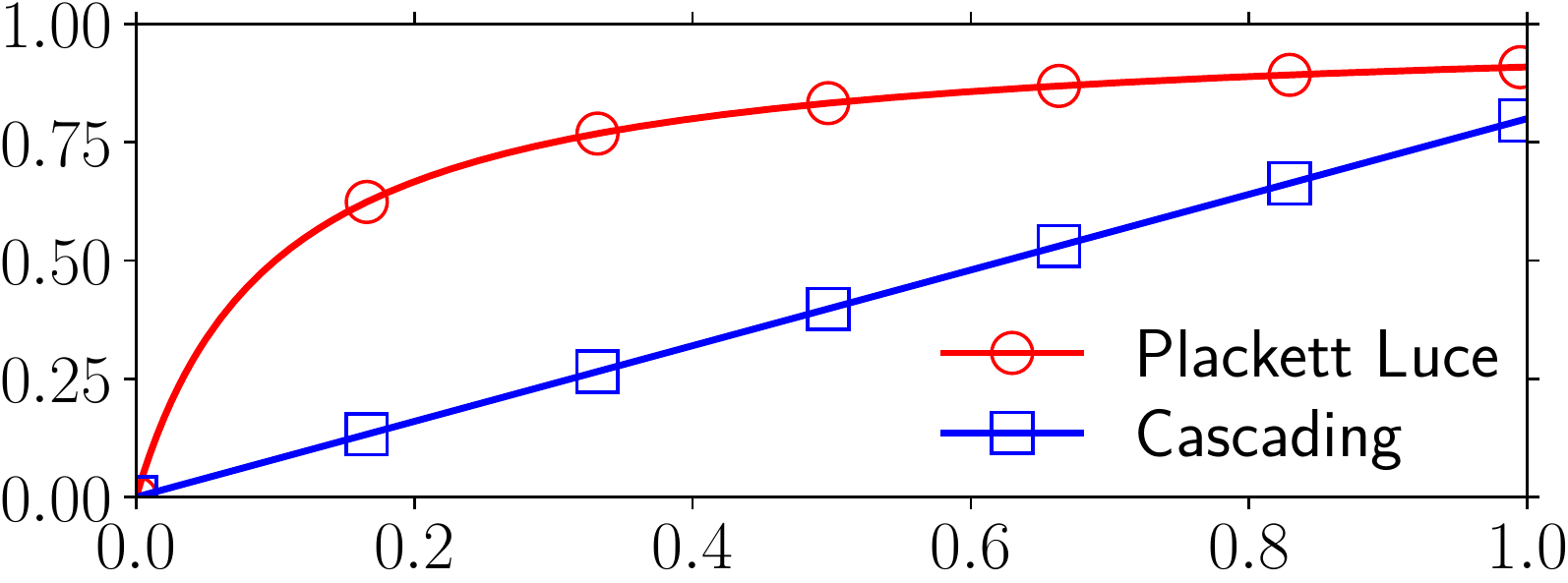}\\
\multicolumn{2}{c}{\hspace{0.5cm}
\small
Item Relevance
}
\end{tabular}
}
\caption{
The relation between click probability and relevance in the Plackett-Luce and cascading scenarios of Section~\ref{sec:cfscenario}.
}
\vspace{-0.5\baselineskip}
\label{fig:cfexample}
\end{figure}

\subsection{Example Scenarios}
\label{sec:cfscenario}

To better understand what the implicit limitations of counterfactual estimation regarding affine click-behavior entail,
we will briefly discuss two models of click behavior that illustrate scenarios where unbiased estimation is possible and impossible.

To start, we will consider a model of behavior where a user first considers all items in a ranking and then chooses according to a Plackett-Luce decision model:
\begin{equation}
P_\text{plackett-luce}(c(d) = 1 \mid y, q) = \frac{\mathds{1}[d \in y] \cdot R_{d|q}}{\sum_{d' \in y} R_{d'|q}}.
\end{equation}
Thus the probability of a user choosing an item in a ranking is equal to its relevance divided by the sum of the relevances of all items in the ranking.
This Plackett-Luce click model is based on well-established decision models from the economic field~\citep{plackett1975analysis, luce2012individual}.
To keep our example simple, we have not added any form of position-bias or trust-bias in the model, although such extensions are certainly possible.

At first glance the Plackett-Luce click model may incorrectly seem to fall within the limitations of counterfactual estimation with the following parameters:
$\alpha_{x(d)} = \mathds{1}[d \in y]/\sum_{d' \in y}  R_{d'|q}$ and $\beta_{x(d)} = 0$ (cf.\  Eq.~\ref{eq:counter:clickprobaffinenorandom}).
However, the denominator of this $\alpha_{x(d)}$ is directly determined by the relevances of \emph{all} items in the ranking, thus also by the item for which relevance is estimated itself.
It therefore does not match the definition of an affine click-model.
This can be clearly seen when plotting how the click probability changes with an item's relevance; Figure~\ref{fig:cfexample} displays this for a scenario where the other relevances that sum up to $0.01$: $\sum_{d' \in y / d} R_{d'|q} = 0.01$.
Clearly, there is not an affine transformation from relevance to click probabilities, and consequently, it is impossible for a click-based counterfactual estimator according to Definition~\ref{def:counterfactualestimate} to provide an unbiased or consistent estimate in this scenario.

This does not mean that unbiased counterfactual estimation is impossible for all click models where click probabilities depend on the relevances of other items.
For example, let us consider a classic cascading click model where a user considers one item at a time until one is clicked~\citep{chuklin2015click, craswell2008experimental}.
The click probability is thus:
\begin{equation}
\begin{split}
&P_\text{casc}(c(d) = 1 \,|\, d, y, q) 
\\
&\hspace{0.5cm} = P_\text{casc}(\forall d' \in y, \, k(d') < k(d) \rightarrow c(d') = 0 \,|\, y, q) \cdot R_{d|q}.
\end{split}
\end{equation}
Therefore. this scenario does provide an affine transformation with $\alpha_{x(d)} = P_\text{casc}(\forall d' \in y, \, k(d') < k(d) \rightarrow c(d') = 0 |\, y, q)$ -- the probability that no earlier item was clicked -- and $\beta_{x(d)} = 0$ (cf.\  Eq.~\ref{eq:counter:clickprobaffinenorandom}).
Figure~\ref{fig:cfexample} also shows how the click probability changes with relevance in a scenario where $\alpha_{x(d)} =0.7$, which clearly reveals an affine transformation.
The crucial difference in this scenario is that the value of $\alpha_{x(d)}$ only depends on other items and not on $d$ itself.

These examples illustrate an intuitive way to think of Theorem~\ref{theorem:cf}:
unbiasedness and consistency guarantees for counterfactual estimation are only possible if there is a linear relation between click probability and item relevance.

\subsection{Adaptive Normalization and Clipping}

While our definition of counterfactual relevance estimates (Definition~\ref{def:counterfactualestimate}) covers almost all click-based counterfactual estimators, it does not fit perfectly with adaptive self-normalization~\citep{swaminathan2015self} and clipping strategies that change with $|\mathcal{D}_q|$~\citep{oosterhuis2021onlinecounterltr, oosterhuis2020topkrankings}.
For instance, a popular clipping strategy is one where $\lim_{|\mathcal{D}_q| \to \infty} \tau = 0$, e.g.\ with $\tau \propto 1/\sqrt{|\mathcal{D}_q|}$~\citep{oosterhuis2021onlinecounterltr, oosterhuis2020topkrankings}.
To incorporate these exceptions one could use the function: 
$f(c_i(d),\allowbreak x_{i}(d),\allowbreak |\mathcal{D}_q|)$
in Definition~\ref{def:counterfactualestimate}.
This change would mean unbiasedness now has to be proven for all possible $|\mathcal{D}_q|$ values, and consistency only requires unbiasedness in the limit, i.e.\ $\lim_{|\mathcal{D}_q| \to \infty} \mathbb{E}[\hat{R}_{d|q}] = R_{d|q}$.
In spite of these differences, Theorem~\ref{theorem:cf} would only need minor changes and the main point would still hold: unbiasedness and consistency are only possible under affine click behavior.
We argue that the simplicity of our theoretical analysis with Definition~\ref{def:counterfactualestimate} and the lack of significant implications of including adaptive strategies, justify our choice of excluding them.

\section{The Limitations of Click Modelling}
\label{sec:limitclickmodel}

The second branch of click-based \ac{LTR} is click-modelling~\citep{chuklin2015click, borisov2016neural, zhuang2021cross, chapelle2009dynamic}; as described in Section~\ref{sec:clickmodel} click modelling methods fit a predictive click model to a dataset of observed user interactions.
Part of this fitted model parameters represent relevance, thus the click model relevance estimates $\hat{R}_{q}$ are variables that are optimized, in contrast with counterfactual estimation where they are the output of the estimator function (see Section~\ref{sec:limitcount}).
We will define a click modelling method by the loss that they optimize:
\begin{definition}
\label{def:clickmodelloss}
A click model loss is a function $\mathcal{L}$ that takes as input the vector of relevance estimates $\hat{R}_{q}$, the estimated latent variables $\hat{Z}$ and the observed data $\mathcal{D}_q$: $\mathcal{L}(\hat{R}_{q},\hat{Z},\mathcal{D}_q)$.
The value of $\mathcal{L}(\hat{R}_{q},\hat{Z},\mathcal{D}_q)$ indicates the quality of both the relevance and latent estimates where minimizing indicates improvement. 
\end{definition}
Our definition is intentionally as generic as possible and - to the best of our knowledge - covers virtually all click modelling methods used for click-based \ac{LTR}~\citep{chuklin2015click, borisov2016neural, zhuang2021cross, chen2020context, chapelle2009dynamic, borisov2016context, borisov2018click, wang2018position}.
Usually, there is a model of user behavior based around the latent variables $\hat{Z}$ and relevances  $\hat{R}_{q}$ and the loss represents the predictive power of that model to describe the data $\mathcal{D}_q$.
In traditional click modelling work, this model is generally a graphical model with a clear interpretation~\citep{chapelle2009dynamic, borisov2016context}, e.g.\ the rank-based position-biased model has a single latent variable per rank representing user examination~\citep{chuklin2015click, craswell2008experimental}.
More recent work has introduced neural networks for click modelling where the latent variables are the parameters of the networks, these models provide no intuitive interpretation but they have enormous predictive power~\citep{borisov2016neural, chen2020context, borisov2018click, zhuang2021cross}.
Our definition covers both traditional and neural click modelling methods by avoiding assumptions about the underlying predictive model.

While in the context of click models, relevance estimates are parameters to be optimized, a click modelling method still has to produce relevance estimates to be used for \ac{LTR}.
We define the output of a click modelling method as the optimal relevance estimates:
\begin{definition}
\label{def:clickmodelestimate}
The optimal relevance estimate of a click model are any relevance estimates that minimize the click model loss:
\begin{equation}
\big(\hat{R}_{q}^*, \hat{Z}^*\big) =
\arg\min_{\hat{R}_{q},\hat{Z}} \mathcal{L}(\hat{R}_{q},\hat{Z},\mathcal{D}_q).
\end{equation}
\end{definition}
We note that this definition is intentionally ambiguous in that it allows for the possibility that multiple relevance estimates could minimize the loss.
Depending on the exact loss and optimization method, it is both possible that there is a single unique vector of optimal estimate values or that multiple values are optimal.

With these definitions, we can now analyse the consistency and unbiasedness of click modelling methods.
To start, we prove the following conditions for consistency:
\begin{theorem}
\label{theorem:cmcons}
A click model estimator is consistent if and only if the only relevance estimates that minimize its loss are the true relevances as $|\mathcal{D}_q|$ tends to infinity:
\begin{equation}
\lim_{|\mathcal{D}_q| \to \infty}\hspace{-0.2cm}
\big(
R_{q} = \hat{R}_q \leftrightarrow 
\min_{\hat{Z}} \mathcal{L}(\hat{R}_{q},\hat{Z},\mathcal{D}_q) = \min_{\hat{R}_q', \hat{Z}} \mathcal{L}(\hat{R}_{q}',\hat{Z},\mathcal{D}_q)
\big).
\label{eq:cmestimatereq}
\end{equation}
\end{theorem}
\begin{proof}
From Definition~\ref{def:consistency} and~\ref{def:clickmodelestimate}, we see that the true relevance estimates need to minimize the click loss in order to be a possible optimal estimate:
\begin{equation}
\lim_{|\mathcal{D}_q| \to \infty}\hspace{-0.2cm}
\big(
R_{q} = \hat{R}_q \rightarrow 
\min_{\hat{Z}} \mathcal{L}(\hat{R}_{q},\hat{Z},\mathcal{D}_q) = \min_{\hat{R}_q', \hat{Z}} \mathcal{L}(\hat{R}_{q}',\hat{Z},\mathcal{D}_q)
\big).
\label{eq:cmconproof1}
\end{equation}
However, Definition~\ref{def:clickmodelestimate} also reveals that multiple relevance estimates may be valid.
Thus to guarantee that the correct values are estimated, we also have to exclude other possible values that minimize the loss:
\begin{equation}
\lim_{|\mathcal{D}_q| \to \infty}\hspace{-0.2cm}
\big(
R_{q} = \hat{R}_q \leftarrow 
\min_{\hat{Z}} \mathcal{L}(\hat{R}_{q},\hat{Z},\mathcal{D}_q) = \min_{\hat{R}_q', \hat{Z}} \mathcal{L}(\hat{R}_{q}',\hat{Z},\mathcal{D}_q)
\big).
\label{eq:cmconproof2}
\end{equation}
Combining Equation~\ref{eq:cmconproof1} and~\ref{eq:cmconproof2} directly proves Theorem~\ref{theorem:cmcons}.
\end{proof}
In other words, Theorem~\ref{theorem:cmcons} proves that a click modelling method is consistent if - in the limit of infinite data - its loss is \emph{only} minimized by the true relevance values.
Moreover, it also proves that these are the only conditions under which it can be consistent.
Importantly, these conditions relate to both the click modelling method, i.e.\ the loss and the underlying predictive model, and the data collection procedure,
since it is both the loss function and the data that determine where minimum values of the loss are.
Thus, while Theorem~\ref{theorem:cmcons} provides a limitation that is conceptually straightforward: the true relevance values should provide the only loss minimum; it appears difficult to understand when this condition is met in practice.

Finally, we prove the following unbiasedness condition:
\begin{theorem}
\label{theorem:unbiascm}
A click modelling method is unbiased if and only if the expected value of its optimal relevance estimates are equal to the true relevances:
\begin{equation}
\mathbb{E}_{\mathcal{D}_q}\big[\hat{R}_{q}^*\big] = R_{q} \longleftrightarrow \forall d, \; \mathbb{E}_{\mathcal{D}_q}\big[\hat{R}_{d|q}^*\big] = R_{d|q}.
\end{equation}
\end{theorem}
\begin{proof}
Follows from the Definition~\ref{def:unbiasedness} and~\ref{def:clickmodelestimate}.
\end{proof}
Admittedly, the condition proved in Theorem~\ref{theorem:unbiascm} is extremely straightforward: the expected value of the estimates that minimize the loss should be equal to the true relevance values.
Similar to the consistency conditions of Theorem~\ref{theorem:cmcons}, this conditions relates to both the loss function of the click model and to the process that gathers the data $\mathcal{D}_q$.
Furthermore, while the unbiasedness condition is also conceptually straightforward, it is again hard to understand in practical terms.

To summarize our theoretical findings on click modelling for click-based \ac{LTR}:
We have provided a very general definition of click modelling methods that covers virtually all existing methods in the field~\citep{chuklin2015click, borisov2016neural, zhuang2021cross, chen2020context, chapelle2009dynamic, borisov2016context, borisov2018click, wang2018position}.
From this definition, we derived the condition for consistency: correct relevance estimates need to provide the only minimum of the loss in the limit of infinite data;
and the condition for unbiasedness: the expected value of the estimates that minimize the loss need to be equal to the correct relevance estimates.
While these conditions are conceptually straightforward, it heavily depends on the exact model and data collection procedure whether they are met.
As a result, it seems there is no straightforward general way to determine whether a click modelling method is consistent or unbiased in practice.

On a higher level, it thus appears that click modelling cannot provide robust guarantees on unbiasedness or consistency, but simultaneously, we were unable to find clear and substantial limitations regarding the click behavior they can debias, as we have found for counterfactual estimation.
Moreover, our theoretical findings contrast with the recent empirical success of this approach that showed it to be an effective approach for click-based \ac{LTR}.
We thus conclude that the main limitation of click modelling methods is that we currently lack robust theoretical guarantees regarding their unbiasedness or consistency.

\subsection{Example Scenario}

We found that the conditions we derived for the unbiasedness and consistency of click modelling methods (Theorem~\ref{theorem:unbiascm} and~\ref{theorem:cmcons}) are conceptually straightforward, but they also heavily depend on the exact model and data collection procedure.
To illustrate this dependency, we will consider the following example scenario that consists of two rankings with the following click probabilities:
\vspace{0.2\baselineskip}
\begin{tabular}{l c c c c | c c c c}
& \multicolumn{4}{c}{Ranking 1} & \multicolumn{4}{c}{Ranking 2} \\
\multicolumn{1}{ l |}{Items} & A & B & C & D & B & A & D & C \\
\multicolumn{1}{ l |}{Click Prob.} & 0.90 & 0.64 & 0.40 & 0.05 & 0.8 & 0.72 & 0.20 & 0.10 
\end{tabular}
Given all the click probabilities, one could fit the commonly used rank-based position bias model~\citep{chuklin2015click, craswell2008experimental} to this scenario: $P(c(d) \mid d, k, q) = \alpha_k R_{d|q}$ with $\alpha_1 = 1$.
Clearly, the result would follow $\alpha_2 = 0.8$, $\alpha_3 = 4\cdot\alpha_4$, $R_{A} = 0.9$, $R_{B} = 0.8$ and $R_{C} = 2\cdot R_{D}$.
Thus, while the values for $R_A$ and $R_B$ are determined exactly, $R_C$ and $R_D$ are merely constrained and a multitude of values would fit the data equally well.
We can thus conclude that this click model does not provide consistent estimates for the relevances in this scenario.

Our example scenario shows the difficulty in determining whether a click model is consistent.
In this scenario, even when all click probabilities and one bias parameter are known, we are unable to determine all document relevances.
Importantly, if we had an additional ranking that placed item C or D at one of the first two ranks this would be possible, illustrating consistency depends both on the model and how data is gathered.
We have not considered unbiasedness because it requires an extensive amount of assumptions and derivation: including a probabilistic model of what rankings are displayed and how many interactions are logged.
Finally, we note that the click model of our example scenario is extremely simple compared to recent neural click models~\citep{borisov2016neural, chen2020context, borisov2018click, zhuang2021cross}, we expect that analysing unbiasedness or consistency conditions  for such complex models is even more infeasible.
 
\subsection{Bias Estimation With Click Modelling}
\label{sec:biasestimationcm}

While our discussion has focused on click modelling for relevance estimation, it is also often applied to estimate bias parameters for counterfactual estimation~\citep{wang2016learning, vardasbi2020trust, oosterhuis2020taking, wang2018position}.
While a full discussion of this application falls outside the scope of this work, we note that our methodology in Section~\ref{sec:limitclickmodel} can be applied analogously to bias estimation.
This would lead to similar a conclusion: for unbiased and consistent estimation of the bias parameters, only the correct vector of values should minimize the loss in expectation and in the limit.
Correspondingly, in spite of its empirical success, it thus seems that click modelling currently lacks any robust theoretical guarantees for bias estimation.

\section{Other Click-Based LTR Methods}
\label{sec:othermethods}

We have attempted to cover as many existing click-based \ac{LTR} methods in our discussion as possible.
Nonetheless, we will now briefly discuss several methods that fall outside of our definitions of counterfactual estimation and click modelling.

Firstly, we have not considered the multitude of online \ac{LTR} methods~\citep{hofmann2013thesis, yue2009interactively, schuth2016mgd, raman2013stable}, because several earlier works have already determined they are not able to provide theoretical guarantees w.r.t.\ unbiasedness~\citep{oosterhuis2020thesis, ai2021unbiased, oosterhuis2021onlinecounterltr, oosterhuis2019optimizing}.
Secondly, \citeauthor{ovaisi2020correcting} propose using Heckman corrections~\citep{heckman1979sample}, commonly used in econometrics, to the click-based \ac{LTR} problem~\citep{ovaisi2021propensity, ovaisi2020correcting}.
Unfortunately, this approach fell outside the scope of this paper due it being so different from the rest of the click-based \ac{LTR} field.
Lastly, there are several click-based pairwise \ac{LTR} works: some were excluded because they do not aim to be unbiased~\citep{wang2021non, joachims2002optimizing}.
A notable exception is the pairwise method by \citet{saito2020unbiased} which does not match our \ac{LTR} problem setting, but does rely on counterfactual relevance estimation, making our analysis in Section~\ref{sec:limitcount} still applicable.
Finally, there is the \emph{Unbiased LambdaMart} method by \citet{hu2019unbiased} which is unbiased but not due to its method, the next section will address this interesting case separately.

\subsection{Unbiased LambdaMart: Trivially Unbiased}

\citet{hu2019unbiased} propose \emph{Pairwise Debiasing} that aims to estimate a pairwise loss, based on the assumption that per rank there is a static ratio between click probabilities and relevances across queries:
\begin{definition}
\emph{Pairwise Debiasing Assumption.}
Per position $k$ there are two ratios $t^{+}_{k} \in \mathbb{R}_{>0}$ and $t^{-}_{k} \in \mathbb{R}_{>0}$ between click and non-click probabilities and relevances when any item is displayed at position $k$:
\begin{align}
P_\text{pair}(C = 1 \mid d, k, q) &= t^{+}_{k} R_{d|q}, \label{eq:pairclickprob}\\
P_\text{pair}(C = 0 \mid d, k, q) &= t^{-}_{k} \big(1-R_{d|q}\big). \label{eq:pairnoclickprob}
\end{align}
\end{definition}
\citet{ai2021unbiased} pointed out that since Eq.~\ref{eq:pairclickprob} is equivalent to the popular rank-based position-bias model~\citep{joachims2017unbiased, wang2018position, craswell2008experimental, chuklin2015click}, the additional Eq.~\ref{eq:pairnoclickprob} seems to conflict with that model.
We will now prove that the pairwise debiasing assumption actually implicitly assumes that there is no bias at all in non-trivial ranking circumstances:
\begin{theorem}
If two items $d_1$ and $d_2$ are both displayed at position $k$ for two different queries $q_1$ and $q_2$ and they have different relevances s.t.\ $R_{d_1|q_1} \neq R_{d_2|q_2}$
then under the pairwise debiasing assumption click probabilities at $k$ are equal to the corresponding relevances:
\begin{align}
P_\text{pair}(C = 1 \mid d, k, q) &= R_{d|q}, \label{eq:pairproof1}\\
P_\text{pair}(C = 0 \mid d, k, q) &= 1 - R_{d|q}. \label{eq:pairproof2}
\end{align}
\end{theorem}
\begin{proof}
Because the click probability (Eq.~\ref{eq:pairclickprob}) and non-click probability (Eq.~\ref{eq:pairnoclickprob}) have to sum to one, the following must hold:
\begin{equation}
t^{+}_{k} R_{d_1|q_1}  = 1 - t^{-}_{k} \big(1-R_{d_1|q_1}\big).
\end{equation}
From this we can express $t^{+}_{k}$ in terms of $t^{-}_{k}$ and $R_{d_1|q_1}$, similarly, this can also be done with $t^{-}_{k}$ and $R_{d_2|q_2}$, therefore:
\begin{equation}
t^{+}_{k}   = \frac{1 - t^{-}_{k} \big(1-R_{d_1|q_1}\big)}{R_{d_1|q_1}}
= \frac{1 - t^{-}_{k} \big(1-R_{d_2|q_2}\big)}{R_{d_2|q_2}}.
\end{equation}
From the latter part of that equality, we can derive:
\begin{equation}
 t^{-}_{k}\big(R_{d_1|q_1}  - R_{d_2|q_2}\big)
= R_{d_1|q_1}  - R_{d_2|q_2} .
\end{equation}
Since $R_{d_1|q_1}  \neq R_{d_2|q_2}$: $t^{-}_{k} = 1$ proving Eq.~\ref{eq:pairproof2}, in turn, since Eq.~\ref{eq:pairclickprob} and Eq.~\ref{eq:pairnoclickprob} sum to one: $t^{+}_{k} = 1$ thus also proving Eq.~\ref{eq:pairproof1}.
\end{proof}
Without discussing the details of \citeauthor{hu2019unbiased}'s \emph{Unbiased LambdaMart} method~\citep{hu2019unbiased}, we can already expect it to be unbiased, since its starting assumption already implicitly entails that clicks are unbiased w.r.t.\ relevance.
Therefore, its unbiasedness is completely trivial from a theoretical perspective, i.e.\ rankings according to \acp{CTR} are also unbiased under the pairwise debiasing assumption.

We would like to point out that there may still be practical value in pairwise debiasing, i.e.\ \citet{hu2019unbiased} report promising real-world performance improvements.
For our discussion regarding implicit limitations, it provides an illustrative example of why one should be critical when choosing or evaluating assumptions and the importance of investigating implicit limitations: for they can make the subsequent theoretical findings trivial and inconsequential.

\section{Discussion and Conclusion:\\ The Future of Click-Based LTR}

We have critically analyzed the theoretical foundations of the main branches of unbiased \ac{LTR} for implicit assumptions and the limitations they entail, and came to the following conclusions:
\begin{enumerate*} [label=(\roman*)]
\item Counterfactual estimation can only unbiasedly and consistently estimate relevance when click probabilities follow affine transformations of relevances.
\item The conditions for click modelling methods to be unbiased and consistent do not translate to robust theoretical guarantees;
we currently lack strong guarantees for complex models, i.e.\ the best performing neural networks~\citep{zhuang2021cross, borisov2016neural}.
\item Assumptions that assert constant ratios between clicks and relevance as well as non-clicks and non-relevance~\citep{hu2019unbiased} implicitly assume that clicks are unbiased w.r.t.\ relevance.
\end{enumerate*}
Our findings have revealed significant limitations of the unbiased \ac{LTR} approach:
\begin{enumerate*} [label=(\roman*)]
\item There is a clear limit on the click behaviors that the current counterfactual estimation approach could possibly correct for.
\item The theoretical assumptions of pairwise debiasing~\citep{hu2019unbiased} are not applicable to biased clicks.
\item The current unbiased \ac{LTR} theory lacks the ability to properly analyze the unbiasedness and consistency of click modelling methods.
\end{enumerate*}
These limitations do not undermine the value of the unbiased \ac{LTR} field: its usefulness and effectiveness is very evident by a multitude empirical results~\citep{zhuang2021cross, wang2016learning, wang2018position, joachims2017unbiased, zhuang2021cross, agarwal2019addressing},
but could help guide future directions and provide valuable lessons to the field.

The implicit limitations that we have uncovered for the existing approach, reveal that in order to find unbiased counterfactual estimation methods that are unbiased and consistent w.r.t.\ non-affine click behavior,
future research should differentiate from our generic definition (Definition~\ref{def:counterfactualestimate}).
Concretely, such methods should not rely on averaging over transformed individual clicks, and as a result, proving unbiasedness would be much more difficult for them.
Furthermore, novel theoretical research is needed for finding theoretical guarantees of unbiasedness and consistency for click modelling methods, in particular, those that utilize complex models.
Our findings reveal that it is not enough that the underlying click model of a method matches the real-world click behavior; whether the method will produce unbiased or consistent relevance estimates also depends on how data is gathered, and ultimately, where the minima of the resulting loss are.
Similarly, theoretical research that explores in what circumstances bias estimation is feasible would be very valuable for further understanding the theoretical guarantees of counterfactual estimation.
Lastly, future research that introduces novel assumptions should critically investigate what these assumptions implicitly entail.
One should actively avoid assumptions that - intentionally or incidentially - make learning from clicks a trivial problem.

Finally, the findings of our critical analysis also provide some important lessons for the field:
Firstly, we should realize that unbiasedness is not always possible.
Accordingly, we should thus not invariably expect nor require it from future work, for this could systematically exclude research that tackles novel problem settings (e.g.\ non-affine click behavior).
Secondly, since unbiasedness may not be a realistic long term goal, the field will likely shift to bias mitigation or partial debiasing as feasible future directions.
Correspondingly, we recommend replacing the term \emph{unbiased \ac{LTR}} with the less-demanding \emph{debiased \ac{LTR}} or the neutral \emph{click-based \ac{LTR}}.

\begin{acks}
This work was supported by the Google Research Scholar Program.
All content represents the opinion of the author, which is not necessarily shared or endorsed by their employers and/or sponsors.
\end{acks}

\balance
\bibliographystyle{ACM-Reference-Format}
\bibliography{references}

\end{document}